\documentclass[11pt,a4paper,english]{article}
\usepackage[utf8]{inputenc}			% Allows the use of UTF-8 characters
\usepackage{makeidx}
\usepackage[pdftex]{graphicx}
\usepackage{mathtools,amsfonts}
\usepackage{array}
\usepackage{multirow, hhline}
\usepackage{paralist}

\usepackage[usenames,dvipsnames,table]{xcolor}
\usepackage{wrapfig}
\usepackage[english]{babel}		% Allows hiphenization of words and other stuff
\usepackage[margin=1in]{geometry}	% Defines margins to 1 inch
\usepackage{lipsum}			% provides filler text \lipsum[n]
\usepackage{comment}
\usepackage{chngpage}			% Allows width to be changed dynamically
\usepackage{amsfonts,amssymb,amsmath, amsthm}
\usepackage{amsopn}
\usepackage{color}

\usepackage{pgfplots}
\usepackage{pgfgantt}
\usepackage{tikz}
\usetikzlibrary{arrows}
\usetikzlibrary{calc}
\usepackage{algpseudocode}
\usepackage[Algorithm]{algorithm}
\usepackage{subcaption}
\usepackage{tabto}
\usepackage{csquotes} %quotations
\usepackage{multirow}
\usepackage{amsthm}
\usepackage[colorlinks,urlcolor=black,citecolor=black,linkcolor=black,menucolor=black]{hyperref}

\newtheorem{theorem}{Theorem}
\newtheorem*{theorem*}{Theorem}

\newtheorem{definition}[theorem]{Definition}

\newtheorem{lemma}[theorem]{Lemma}

\newtheorem{observation}[theorem]{Observation}

\newcommand\numberthis{\addtocounter{equation}{1}\tag{\theequation}}

\newcommand{\ra}{\mathit{rank}}

\newcommand{\Y}{\mathbf{Y}}

\title{Fair and Efficient Allocations under Subadditive Valuations}
\author{Bhaskar Ray Chaudhury\thanks{MPI for Informatics, Saarland Informatics Campus, Graduate School of Computer Science, Saarbr\"ucken, Germany}\\ \texttt{\small braycha@mpi-inf.mpg.de} \and Jugal Garg\thanks{University of Illinois at Urbana-Champaign}\\ \texttt{\small jugal@illinois.edu}  \and Ruta Mehta\thanks{University of Illinois at Urbana-Champaign}\\ \texttt{\small rutameht@illinois.edu}}

\date{}
\begin{document}
\maketitle

\begin{abstract}
We study the problem of allocating a set of indivisible goods among agents with subadditive valuations in a fair and efficient manner. Envy-Freeness up to any good (EFX) is the most compelling notion of fairness in the context of indivisible goods. Although the existence of EFX is not known beyond the simple case of two agents with subadditive valuations, some good approximations of EFX are known to exist, namely $\tfrac{1}{2}$-EFX allocation~\cite{TimPlaut18} and EFX allocations with bounded charity~\cite{CKMS20}.

Nash welfare (the geometric mean of agents' valuations) is one of the most commonly used measures of efficiency. In case of additive valuations, an allocation that maximizes Nash welfare also satisfies fairness properties like Envy-Free up to one good (EF1). Although there is substantial work on approximating Nash welfare when agents have additive valuations, very little is known when agents have subadditive valuations. %The best-known approximation guarantee under subadditive valuations is $\mathcal{O}(m)$. 
In this paper, we design a polynomial-time algorithm that outputs an allocation that satisfies either of the two approximations of EFX as well as achieves an $\mathcal{O}(n)$ approximation to the Nash welfare. Our result also improves the current best-known approximation of $\mathcal{O}(n \log n)$~\cite{GargKK20} and $\mathcal{O}(m)$~\cite{NguyenR14} to Nash welfare when agents have submodular and subadditive valuations, respectively. 

Furthermore, our technique also gives an $\mathcal{O}(n)$ approximation to a family of welfare measures, $p$-mean of valuations for $p\in (-\infty, 1]$, thereby also matching asymptotically the current best approximation ratio for special cases like $p =-\infty$~\cite{KhotPonuswami07} while also retaining the remarkable fairness properties.  
\end{abstract}

\section{Introduction}
\label{introduction}
Discrete fair division of resources is a fundamental problem in various multi-agent settings, where the goal is to partition a set $M$ of $m$ \emph{indivisible} goods among $n$ agents in a \emph{fair} and \emph{efficient} manner. Each agent $i$ has a valuation function $v_i: 2^M\rightarrow \mathbb{R}_{\geq 0}$ that quantifies the amount of utility $i$ derives from every subset of goods. We assume that $v_i$'s are \emph{monotone}, i.e., $v_i(A)\leq v_i(A \cup \left\{ g\right\} )$ for all $g \in M$, \emph{normalized} i.e., $v_i(\emptyset) = 0$ and \emph{subadditive}, i.e., $v_i(A\cup B) \leq v_i(A) + v_i(B)$, for all $A, B\subseteq M$. Subadditive functions naturally arise in practice because they capture the notion of \emph{complement-freeness}~\cite{LehmannLN06}. Furthermore, they strictly contain \emph{submodular} functions\footnote{A function $v(.)$ is submodular if $v(A) + v(B) \geq v(A\cup B) + v(A\cap B), \forall A, B\subseteq M$.}, which capture the notion of diminishing marginal returns. 

Among various choices, \emph{envy-freeness} is the most natural fairness concept, where no agent $i$ envies another agent $j$'s bundle, i.e., partition of goods into $n$ bundles $X_1, X_2, \dots, X_n$ so that for all agents $i$ and $j$, we have $v_i(X_i)\ge v_i(X_j)$. However, envy-free allocation do not always exist, e.g., consider allocating a single valuable good among two agents. Its mild relaxation envy-freeness up to \emph{any} good (EFX)~\cite{CaragiannisKMP016} is arguably the most compelling notion of fairness in discrete setting, where no agent envies other's allocation after the removal of \emph{any} good, i.e., for all agents $i$ and $j$, we have $v_i(X_i) \ge v_i(X_j\setminus  \left\{ g \right\})$ for all $g\in X_j$. While it is not known whether an EFX allocation always exists or not beyond the simple case of two agents under subadditive valuations, the following relaxations exist: 
\begin{itemize}
	\item $\tfrac{1}{2}$-EFX allocation $X = \langle X_1,X_2, \dots ,X_n \rangle$ where $v_i(X_i) \ge \tfrac{1}{2} \cdot v_i(X_j\setminus  \left\{ g \right\})$, for all $ g\in X_j$~\cite{TimPlaut18}. In this paper we will be referring to a relaxed version of $\tfrac{1}{2}$-EFX namely, $(\tfrac{1}{2} - \varepsilon)$-EFX allocation where $v_i(X_i) \ge (\tfrac{1}{2} - \varepsilon) \cdot v_i(X_j\setminus  \left\{ g \right\})$  for all $g\in X_j$. A $(\tfrac{1}{2} - \varepsilon)$-EFX allocation can also be computed in polynomial time when agents have subadditive valuations.
	\item EFX allocation with \emph{bounded charity} $X = \langle X_1,X_2, \dots ,X_n \rangle$ where we do not allocate a set $P$ of goods (set $P$ is donated to charity) where $|P|<n$ and $v_i(X_i) \ge v_i(P)$,for all $i \in [n]$ and the partial allocation $X$ is EFX~\cite{CKMS20}. There is also a polynomial time algorithm to find an $(1-\varepsilon)$-EFX allocation with bounded charity for general valuations for any $\varepsilon>0$~\cite{CKMSarxiv}\footnote{This is an updated version of the paper which goes beyond the preliminary version published in SODA 2020}.
\end{itemize}

Another popular (and stronger) relaxation is  envy-freeness up to \emph{one} good (EF1)~\cite{budish2011combinatorial}, where no agent envies other's allocation after the removal of \emph{some} good from the other's bundle, i.e., $v_i(X_i) \ge v_i(X_j\setminus  \left\{ g \right\})$, for \emph{some} $g\in X_j$. Clearly, EFX implies EF1. Although the existence of EFX allocations still remains a major open question, an EF1 allocation always exists for general valuations and can be obtained in polynomial time~\cite{LiptonMMS04}. 

We note that none of the above algorithms provides, as such, any efficiency guarantees. For efficiency, among many choices, maximum Nash welfare, defined as the geometric mean of agents' valuations, serves as a focal point. In contrast to other popular welfare measures such as social welfare and max-min welfare, Nash welfare is scale invariant, i.e., scaling one agent's valuation by any positive constant does not change the outcome. In case of additive valuations\footnote{A valuation function $v(.)$ is additive if $v_i(S) = \sum_{j\in S} v_i(\{j\}), \forall S$.}, an allocation that maximizes Nash welfare is both EF1 and Pareto optimal\footnote{An allocation $X'=(X'_1, \dots, X'_n)$ Pareto dominates another allocation $X=(X_1, \dots, X_n)$ if $v_i(X'_i)\geq v_i(X_i),\forall i$ and $v_k(X'_k)>v_k(X_k)$ for some $k$. An allocation $X$ is Pareto optimal if no allocation $X'$ dominates $X$.}~\cite{CaragiannisKMP016}. However, such an allocation does not provide the EF1 property beyond additive (e.g., subadditive valuations~\cite{CaragiannisKMP016}), and further, no meaningful guarantee in terms of EFX even in the case of additive valuations~\cite{ABFHV20}. Furthermore, maximizing the Nash welfare is a hard problem, and the best known approximation guarantees are $\mathcal{O}(n\log{n})$ and $\mathcal{O}(m)$ for submodular~\cite{GargKK20} and subadditive~\cite{NguyenR14} valuations, respectively. As the case with the algorithms providing fairness guarantees, these Nash welfare approximation algorithms do not provide any fairness guarantees. Therefore, a natural question is: 

\begin{quote}
	\begin{center}
		\emph{Does there exist a polynomial-time algorithm that provides the best known fairness guarantees as well as the best known efficiency guarantees simultaneously?}
	\end{center}
\end{quote}

In this paper, we answer this question affirmatively. We design a simple algorithm that outputs an allocation that provides $(i)$ either of the best-known EFX approximations mentioned above, $(ii)$ EF1 guarantee, and $(iii)$ $\mathcal{O}(n)$ approximation to the maximum Nash welfare. The latter also improves the best-known approximation factor. Further, we show that our algorithm can be easily adapted to obtain the same guarantees for the entire family of $p$-mean welfare measures $M_p(X)$, defined as, \[M_p(X) = (\sum_{i} \tfrac{1}{n}(v_i(X_i))^p)^{1/p} \ \ \text{  for  } p\in (-\infty, 1].\] The $p=-\infty, 0$, and $1$ correspond to the well-studied cases of max-min welfare, Nash welfare, and social welfare, respectively. We note that this also matches the current best approximation ratio for the max-min welfare~\cite{KhotPonuswami07} while also retaining the above mentioned fairness guarantees.

One crucial difference between Nash welfare and $p$-mean welfare when $p \neq 0$  is that $p$-mean is no longer scale invariant. Therefore, it is not intuitive that the allocation that maximizes welfare will be fair.\footnote{Consider a special case when $p = -\infty$. Here the $p$-mean welfare is equal to the valuation of the agent with smallest valuation. In particular, consider the scenario with two agents and $n$ goods where agent 1 has a valuation of $1$ for each good and agent 2 has a valuation of $\varepsilon < \tfrac{1}{n}$ for each good. The allocation that maximizes the $p$-mean welfare here will give exactly one good to agent 1 and $n-1$ goods to agent 2, which is very far from satisfying any relaxation of envy-freeness.} However, we manage to give a polynomial time algorithm that achieves a good approximation (independent of the number of goods in the instance) to the $p$-mean welfare while still retaining all the fairness properties.

\subsection{Technical Overview}\label{overview}
In this section, we briefly sketch our main result and overall approach. One of the primary motivation of our techniques is to show that finding certain ``fair allocations'' can give us ``reasonably efficient allocations''. While any arbitrary EF1 allocation does not give us any guarantee on the generalized $p$-mean welfare even in the context of additive valuations, we would outline in this paper that certain EFX allocations (need not be Pareto-optimal) and even approximate EFX allocations can help us get good approximations to a broad class of welfare measures like the generalized $p$-mean welfare, further showing the strength of this measure of fairness. We now state the main result of our paper

\begin{theorem}
	Given a discrete fair division instance with a set $[n]$ of $n$ agents, a set $M$ of $m$ indivisible goods, where each agent $i$ has a subadditive valuation function $v_i: 2^M \rightarrow \mathbb{R}_{\geq 0}$, for any $\varepsilon > 0$,  we can find in polynomial time 
	\begin{itemize}
		\item  a partition $\langle X_1,X_2, \dots ,X_n \rangle$ of $M$ such that $X$ is $(\tfrac{1}{2} - \varepsilon)$-EFX and $M_p(X) \geq \tfrac{1- 2\varepsilon}{8(n+1)} \cdot M_p(X^*)$, and 
		\item a partition $\langle X_1,X_2, \dots ,X_n,P \rangle$ of $M$ such that $X$ is $(1 - \varepsilon)$-EFX with bounded charity and $M_p(X) \geq \tfrac{1-\varepsilon}{4(n+1)} \cdot M_p(X^*)$,
	\end{itemize}
	where $X^*$ is the allocation with maximum $p$-mean value.
\end{theorem}
We now briefly sketch our main techniques: Let us consider the scenario that a given instance admits an envy-free allocation, i.e., a partition of the goods into $n$ bundles $X_1,X_2, \dots ,X_n$ such that for all pairs of agents $i$ and $j$ we have $v_i(X_i) \geq v_i(X_j)$. In that case for each agent $i$ we have
\begin{align*}
n \cdot v_i(X_i) &\geq \sum_{j \in [n]} v_i(X_j)\\
&\geq v_i(\cup_{j \in [n]} X_j) &(\text{by subadditivity})\\
&=v_i(M)
\end{align*}
This implies that $v_i(X_i) \geq \tfrac{1}{n} \cdot v_i(M)$. Since in any optimal allocation no agent can get a valuation more than $v_i(M)$, we can conclude that each agent has a bundle worth $\tfrac{1}{n}$ times his bundle at optimum. This would immediately give us an $n$ approximation for generalized $p$-mean welfare. However, most instances may not admit an envy-free allocation. Naturally, we then look into the closest relaxation of envy-freeness that is known to exist in the context of indivisible goods\footnote{In our algorithm we consider relaxed variants of these notions like $(\tfrac{1}{2} - \varepsilon)$-EFX and $(1-\varepsilon)$-EFX with bounded charity, but for clarity in this section we keep the original notions}. ($\tfrac{1}{2}$-EFX ~\cite{TimPlaut18}  and EFX with \emph{bounded charity}\footnote{defined in Section~\ref{preliminaries}}~\cite{CKMS20}). So let us consider the $\tfrac{1}{2}$-EFX allocation: Here we can partition the given instance into $n$ bundles $X_1,X_2,\dots ,X_n$ such that for all pairs of agents $i$ and $j$ we have $v_i(X_i) \geq \tfrac{1}{2} v_i(X_j \setminus \left\{g\right\})$ for all $g \in X_j$. Let us first look into all the bundles $X_j$ that are not singleton, i.e., $\lvert X_j \rvert \geq 2$: We have that $v_i(X_i) \geq \tfrac{1}{2} \cdot v_i(X_j \setminus \left\{g\right\})$ for \emph{all} $g \in X_j$, implying that $v_i(X_i) \geq \tfrac{1}{2} \cdot \mathit{max}\Big(v_i(X_j \setminus \left\{g \right\}),v_i(\left\{g \right\})\Big)$ (as $\lvert X_j \rvert \geq 2$). Thus,
\begin{align*}
n \cdot v_i(X_i) &\geq \sum_{\lvert X_j \rvert \geq 2}\frac{1}{2} \cdot \frac{1}{2} \cdot \Big(v_i(X_j \setminus \left\{g \right\}) + v_i(\left\{g \right\}) \Big)\\
&\geq \frac{1}{4} \cdot \sum_{\lvert X_j \rvert \geq 2} v_i(X_j) &(\text{by subadditivity})\\
&\geq \frac{1}{4} \cdot v_i(\cup_{\lvert X_j \rvert \geq 2} X_j ) &(\text{by subadditivity}) \numberthis \label{eqtechnicaloverview}
\end{align*} 

Let $S$ be the set of all the goods in singleton bundles in $X$, i.e., $S = \left\{g \mid \text{there is a $X_j = \{g\}$} \right\}$. Then from~(\ref{eqtechnicaloverview}) we have the guarantee that for every agent $v_i(X_i) \geq \tfrac{1}{4n} \cdot v_i(M \setminus S)$. Therefore, in any $\tfrac{1}{2}$-EFX allocation every agent has an $\tfrac{1}{4n}$ fraction of his valuation on the goods he receives from $M \setminus S$ in the optimal allocation, i.e., $v_i \big(X^*_i \cap (M \setminus S) \big)$ where $X^*=(X^*_1,\dots,X^*_n)$ is the allocation that has the highest generalized $p$-mean welfare. The only problem is how to allocate the goods in the set $S$ appropriately. 

The only scenario where an incorrect allocation of the goods in $S$ causes a significant decrease in the $p$-mean welfare is when there are agents who have a substantially high valuation for some goods in $S$. However, we could very well be in a scenario where there are only a few goods in $S$ (say less than $\tfrac{n}{3}$) which are very valuable to many agents and then we may not be able to give every agent a bundle that he values $\tfrac{1}{n}$ times the whole set $S$.\footnote{A very simple scenario is to divide $n$ goods among $n$ agents with identical additive valuations, where all agents have a valuation of $1$ for a single good and $\varepsilon \ll \tfrac{1}{n}$ for the rest of the goods. In any division there will be $n-1$ agents who do not get $\tfrac{1}{n}$ of their valuation on the set of $n$ goods} Therefore we need to compare our allocation with the allocation that maximizes the $p$-mean welfare. 
%\RM{we can not here that even in the optimal solution all these goods can not be allocated to all the agents, and hence gives a better comparision.}

We briefly sketch how we overcome this barrier. The good aspect of the situation is that the number of goods in $S$ are small, i.e., $\lvert S \rvert \leq n$.  Let $H_i$ denote the set of $n$ goods that are valued by agent $i$ the most, {\em i.e.}, all goods in $H_i$ are more valuable than any good outside $H_i$. Now we find a \emph{single good allocation} (where each agent gets exactly one good) of the \emph{high valued goods}, namely the set $\mathbf{H}=\cup_{i \in [n]} H_i$, optimally to the agents assuming that we can give each agent at least $\tfrac{1}{n}$ times their valuation for the \emph{low valued goods}, namely the set $M \setminus H_i$. That is, we find a single good allocation, where every agent $i$ gets exactly one high valued good $h_i\in H_i$, that maximizes $\sum_{i \in [n]} \big( v_i(\left\{h_i \right\}) + \tfrac{1}{n} v_i(M \setminus H_i) \big)^p$ (such allocations can be found efficiently with matching algorithms). Let us call the current single good allocation $Y$. Note that $Y$ is trivially EFX as every agent has exactly one good (therefore $Y$ is also $\tfrac{1}{2}$-EFX). We then run the $\tfrac{1}{2}$-EFX algorithm starting with $Y$ as the initial partial $\tfrac{1}{2}$-EFX allocation. \emph{The intuition being that the low valued goods appear in non-singleton bundles and the high valued goods occur in singleton bundles} in the final $\tfrac{1}{2}$-EFX allocation, but we have allocated the high valued goods correctly (up to a factor of $\tfrac{1}{n}$ as we computed a single good allocation, while the optimum need not necessarily give every agent exactly one high valued good) as we started out with an optimal allocation of the high valued goods. Since the low valued goods occur in non-singleton bundles we are indeed able to give every agent $\tfrac{1}{n}$ times their valuation for the low valued goods.

\subsection{Related Work}\label{sec:rel-work}
Fair division has been extensively studied for more than seventy years since the seminal work of Steinhaus~\cite{Steinhaus48}. A complete survey of all different settings and the fairness and efficiency notions used is well beyond the scope of this paper. We limit our attention to the discrete setting (as mentioned in Section~\ref{introduction}) and the two most universal notions of fairness, namely  \emph{envy-freeness} and \emph{proportionality}\footnote{In a proportional share, each agent receives at least a $1/n$ share of the goods.}. Both of these properties can be guaranteed in case of divisible goods. For indivisible goods, there are trivial instances (mentioned in Section~\ref{introduction}) where neither of these notions can be achieved by any allocation. However there has been extensive studies on relaxations of envy-freeness like EF1~\cite{BudishCKO17,BarmanBMN18,LiptonMMS04,CaragiannisKMP016} and EFX~\cite{CKMS20,CaragiannisGravin19,CaragiannisKMP016,TimPlaut18} and relaxations of proportionality like maximin shares (MMS)~\cite{budish2011combinatorial,BL16,AMNS17,BK17,KPW18,GhodsiHSSY18,JGargMT19,GargT19} and proportionality up to one good (PROP1)~\cite{ConitzerFS17,BarmanK19,GargM19}. 

While there is a significant interest in  finding fair allocations, there has also been a lot of interest in guaranteeing \emph{efficient} fair allocations. A common measure of efficiency in economics is \emph{Pareto-Optimality}\footnote{Defined in Section~\ref{introduction}}. Caragiannis et al. ~\cite{CaragiannisKMP016} showed that any allocation that maximizes the Nash welfare is guaranteed to be Pareto-optimal (efficient) and EF1 (fair). Hence the Nash welfare in itself also a good measure of efficiency of fair allocations. Unfortunately, finding an allocation with the maximum Nash welfare is APX-hard~\cite{Lee17}. However, approximation algorithms for Nash welfare under different types of valuations have received significant attention recently, e.g.,~\cite{ColeG18,ColeDGJMVY17,AnariGSS17,GargHM18,AnariMGV18,BKV18,ChaudhuryCGGHM18,GargKK20}. 

\subsection{Independent Work}
Independently of our work, Barman et al.~\cite{BarmanBKS'20} also find an $\mathcal{O}(n)$-approximation algorithm for the generalized $p$-mean welfare when agents have subadditive valuations. On a high level, both algorithms, first carefully allocate a single \emph{highly valuable} good to each agent and then carefully allocate the remaining goods. However, the procedures used to determine the initial (the single highly valuable good allocation) and the final allocations are significantly different. Also, contrary to the allocation determined by the algorithm in~\cite{BarmanBKS'20}, we are able to obtain guarantees on the fairness of our allocation, namely the properties of EF1 and the two relaxations of EFX. 

In the same paper, Barman et al.~\cite{BarmanBKS'20} show that it requires an exponential number of value queries to provide any sublinear approximation for the generalized $p$-mean welfare under subadditive valuations. Therefore, in polynomial time, our algorithm yields an allocation that satisfies the best relaxations of EFX known for subadditive valuations, and also achieves the best approximation for the generalized $p$-mean welfare possible in polynomial time (assuming $\mathbf{P} \neq \mathbf{NP}$).
\section{Preliminaries}
\label{preliminaries}
Any discrete fair division instance $I$ is a tuple $\langle [n], M, \mathcal{V} \rangle$ comprising of a set of $n$ agents ($[n]$), a set of $m$ goods ($M$) and a set of valuation functions $\mathcal{V} =\left\{v_1,v_2, \dots ,v_n \right\}$. The valuation function $v_i: 2^M \rightarrow \mathbb{R}_{\geq 0}$ tries to capture agent $i$'s utility for each subset of goods. Throughout this paper we will be dealing with the case where all the valuation functions are 
\begin{itemize}
	\item \emph{normalized}: $v_i(\emptyset) = 0$ for all $i \in [n]$,
	\item \emph{monotone}: $v_i(A \cup \left\{g\right\}) \geq v_i(A)$ for all $i \in [n]$ and $A \subset M$, and 
	\item \emph{subadditive}: for any sets $A,B \subseteq M$ we have $v_i(A) + v_i(B) \geq v_i(A \cup B)$ for all $i \in [n]$.
\end{itemize}

For ease of notation we use $v_i(g)$ instead of $v_i(\left\{g\right\})$ and $v_i(A \setminus g )$ instead of $v_i(A \setminus \left\{g\right\})$

\paragraph{Generalized $p$-mean welfare:} Given an allocation $X$ the $p$-mean welfare of $X$ (parametrized by $p$) is defined by
\begin{align*}
M_p(X) &= \Big( \frac{1}{n} \sum_{i \in [n]} v_i(X_i)^p \Big)^{\tfrac{1}{p}}
\end{align*}

This captures a wide range of fairness and efficiency measures that have been used frequently in the literature: Nash welfare for $p=0$, max-min welfare (also known as the egalitarian welfare) for $p = -\infty$ and social welfare for $p=1$. Barman and Sundaram~\cite{BarmanRanjani_pmeans} also mention that,

\begin{quote}
	``\textit{generalized means with $p \in (-\infty,1]$ exactly constitute the family of welfare functions that satisfy the Pigou-Dalton transfer principle and a few other key axioms.}'' 
\end{quote}
In the same paper they show that when agents have identical valuations, there is an algorithm that provides an $\mathcal{O}(1)$ factor approximation to the $p$-mean welfare.

\paragraph{EFX Allocations and Relaxations:} EFX is arguably the strongest notion of fairness in the context of indivisible goods. Formally,

\begin{definition}
	An allocation $X = \langle X_1,X_2, \dots ,X_n \rangle$ is said to be an EFX allocation if for all pairs of agents $i$ and $j$, we have $v_i(X_i) \geq v_i(X_j \setminus g)$ for all $g \in X_j$. 
\end{definition} 

Although the existence of complete EFX allocations is not known yet, there have been results pertaining to the existence of certain relaxations of EFX. We state two major relaxations here. Plaut and Roughgarden~\cite{TimPlaut18} introduced the notion of \emph{approximate EFX} or equivalently $\alpha$-EFX: 

\begin{definition}
	An allocation $X$ is $\alpha$-EFX with $\alpha \in (0,1)$ if and only if for all pairs of agents $i$ and $j$ we have $v_i(X_i) \geq \alpha \cdot v_i(X_j \setminus g)$ for all $g \in X_j$. 
\end{definition}
Plaut and Roughgarden~\cite{TimPlaut18} showed that $\tfrac{1}{2}$-EFX allocations exist and can be computed in pseudo-polynomial time. With a very minor change in their algorithm\footnote{Just run the same algorithm replacing the violation condition from $\tfrac{1}{2}$-EFX to $(\tfrac{1}{2} - \varepsilon)$-EFX} we can obtain an $(\tfrac{1}{2}- \varepsilon)$-EFX allocation in polynomial time. 

Another relaxation introduced by Chaudhury et al.~\cite{CKMS20} is \emph{EFX with bounded charity}:
\begin{definition}
	\label{EFXboundedcharity}
	A partial allocation $X$ is an EFX allocation with bounded charity with the set of unallocated goods $P$ such that 
	\begin{itemize}
		\item $X$ is EFX,
		\item $\lvert P \rvert < n$, and 
		\item $v_i(X_i) \geq v_i(P)$ for all $i \in [n]$.
	\end{itemize}
\end{definition}

The updated version of the paper~\cite{CKMSarxiv} gives a polynomial time algorithm to determine $(1- \varepsilon)$-EFX allocation with bounded charity.\footnote{Just relax the first condition in Definition~\ref{EFXboundedcharity} to ``$X$ is $(1 - \varepsilon)$-EFX"} 

Throughout the paper we will outline algorithms that find allocations with high welfare and are flexible with the fairness that the allocations satisfy, i.e., we can get $(\tfrac{1}{2}-\varepsilon)$-EFX allocations with high welfare  or $(1-\varepsilon)$-EFX allocations with bounded charity and high welfare. Therefore we now introduce some common notation for ease of referring to both these relaxations of EFX at the same time: 

\begin{definition}
	An allocation $X$ is an $(\alpha,c)$-EFX allocation with $\alpha \in (0,1)$ and $c \in \left\{0,1\right\}$ if and only if 
	\begin{itemize}
		\item $X$ is $\alpha$-EFX and EF1,
		\item $\lvert P \rvert < n$, and 
		\item $v_i(X_i) \geq v_i(P)$ for all $i \in [n]$.
		\item $P = \emptyset$ if $c=1$.\footnote{$c$ serves as an indicator to whether the allocation is complete or not.}
	\end{itemize}
\end{definition} 

Therefore an $(\alpha,1)$-EFX allocation would refer to an $\alpha$-EFX (which is also EF1) allocation and a $(\alpha,0)$-EFX allocation would refer to an $\alpha$-EFX allocation with bounded charity. In particular we would only be interested in $((\tfrac{1}{2} - \varepsilon),1)$-EFX allocation and $(1-\varepsilon,0)$-EFX allocation.

Similarly, we also introduce the notion of an $(\alpha,c)$-EFX algorithm:

\begin{definition}
	An $(\alpha,c)$-EFX algorithm takes as input any partial $\alpha$-EFX allocation $X$ and outputs an $(\alpha,c)$-EFX allocation $Y$ as the final allocation such that 
	\begin{itemize}
		\item the valuation of every agent in the final valuation is at least as high as his valuation in the initial allocation, i.e., $v_i(Y_i) \geq v_i(X_i)$, and 
		\item if there exists an agent $i$ such that $\lvert Y_i \rvert =1$ and $Y_i \neq X_i$, then $v_i(Y_i) > v_i(X_i)$.
	\end{itemize}
\end{definition}

In particular an $(\alpha,c)$-EFX algorithm outputs a final $(\alpha,c)$-EFX allocation that preserves (if not improves) all the  welfare guarantees of the initial $\alpha$-EFX allocation. Both the existing algorithms for determining an $(\tfrac{1}{2} - \varepsilon,1)$-EFX  allocation (a trivial modification of the algorithm in ~\cite{TimPlaut18}) and $(1- \varepsilon,0)$-EFX ~\cite{CKMSarxiv} allocation are an $(\tfrac{1}{2}-\varepsilon,1)$-EFX algorithm and an $(1-\varepsilon,0)$-EFX algorithm respectively.

\section{Algorithm}

In this section, we show that we can determine an $(\alpha,c)$-EFX allocation with an $\mathcal{O}(n)$ approximation on the $p$-mean welfare. %Note that $p=0$ corresponds to the Nash welfare. One crucial difference between Nash welfare and $p$-mean welfare when $p \neq 0$  is that $p$-mean is no longer scale invariant. Therefore it is not intuitive that the allocation that maximizes welfare will be fair.\footnote{Consider a special case when $p = -\infty$. Here the $p$-mean welfare is equal to the valuation of the agent with smallest valuation. In particular, consider the scenario with two agents and $n$ goods where agent 1 has a valuation of $1$ for each good and agent 2 has a valuation of $\varepsilon$ for each good. The allocation that maximizes the $p$-mean welfare here will give exactly one good to agent 1 and $n-1$ goods to agent 2, which is very far from satisfying any relaxation of envy-freeness.} However, we manage to give a polynomial time algorithm that achieves a good approximation (independent of the number of goods in the instance) to the $p$-mean welfare while still retaining all the fairness properties.
The algorithm is very simple and it has just two phases: In the first phase we try to determine a reasonable allocation of high valued goods (we call this allocation $Y$) and then in the second phase we just run an $(\alpha,c)$-EFX algorithm with the remaining set of goods (we call our final allocation $Z$).

\paragraph{Allocating the high valued goods $Y$:} We first formally define the notion of high valued goods for an agent: For each agent $i$ let us order the goods in $M$ as $\left\{ g^{i}_1,g^{i}_2, \dots , g^{i}_m  \right\}$ such that $v_i(g^{i}_1) \geq  v_i(g^{i}_2) \geq \dots \geq v_i(g^{i}_m)$. Let $H_i = \left\{ g^{i}_1,g^{i}_2, \dots , g^{i}_n \right\}$. We refer to $H_i$ to be the set of high valued goods for agent $i$. Also for each good $g^{i}_k$ and an agent $i$ we define $\ra_i(g^{i}_k) = k$. Notice that if for any agent $i$ $\ra_i(g) < \ra_i(g')$, then $v_i(g) \geq v_i(g')$. 

We now outline how we compute the initial allocation $Y$. We construct the complete bipartite graph $G = ([n] \cup M, [n] \times M)$ with the weight of the edge from agent $i$ to good $g$, $w_{ig}$ being 
\begin{itemize}
	\item $n \cdot v_i(g) +  v_i (M \setminus H_i) $ if $p = - \infty$,
	\item $\log \Big(n \cdot  v_i(g) + v_i(M \setminus H_i) \Big)$ if $p=0$ and  
	\item $\Big(n \cdot v_i(g) + v_i (M \setminus H_i) \Big)^p$ otherwise.
\end{itemize}

Depending on the value of $p$ we choose an appropriate matching mechanism to determine $Y$: $Y$ is determined such that $\cup_{i \in [n]}(i,Y_i)$ is 
\begin{itemize}
	\item a maximum weight matching in $G$   if $p \geq 0$,
	\item a minimum weight perfect matching in $G$ if $p <0 $ and $p \neq -\infty$,
	\item a max-min matching\footnote{This is a matching that maximizes the weight of the smallest edge in the matching.} in $G$ if $p = -\infty$.
\end{itemize}

Let $Y$ be the allocation outputed by the corresponding matching subroutine. Also let $\mathbf{Y} = \cup_{i \in [n]} Y_i$. We modify the allocation $Y$ slightly such that $\cup_{i \in [n]}(i,Y_i)$ still remains the optimum matching, but no agent will prefer a good outside $\mathbf{Y}$ to the good allocated to him in $Y$ ($Y_i$), i.e., we wish to determine an allocation $Y$ such that for all agent $i \in [n]$ and all $g \notin \mathbf{Y}$ we have that $\ra_i(Y_i) < \ra_i(g)$. To achieve this, as long as there is an agent $i \in [n]$ and a good $g \notin \mathbf{Y}$ such that $\ra_i(g) < \ra_i(Y_i)$ we set $Y_i \gets \left\{g\right\}$. Note that such an operation does not decrease the optimum value of the matching: $v_{i}(g) \geq v_{i}(Y_i)$ (as $\ra_i(g) < \ra_i(Y_i)$) and hence $w_{ig} \geq w_{iY_i}$ for $p = -\infty$ and $p \in [0,1]$, while $w_{ig} \leq w_{iY_i}$ for $p <0$ and $p \neq -\infty$. This implies that the objective value of the matching does not decrease when $p \in [0,1]$ and $p = -\infty$ and the objective value of the matching does not increase when $p < 0$ and $p \neq -\infty$. Therefore, $\cup_{i \in [n]} (i,Y_i)$ still stays an optimum matching, but $\sum_{i \in [n]} \ra_i(Y_i)$ strictly decreases. Since $n \leq \sum_{i \in [n]} \ra_i(Y_i) \leq nm$, after $\mathcal{O}(nm)$ iterations we will have an allocation $Y$ such that $\cup_{i \in [n]}(i,Y_i)$ is still an optimum matching, but for all agents $i \in [n]$ and for all goods $g \notin \mathbf{Y}$ we have $\ra_i(Y_i) < \ra_i(g)$.

The complete algorithm is outlined in Algorithm~\ref{algorithmpmeans} (Selection of the allocation $Y$ is captured in steps 1 to 5).

\begin{lemma}\label{Yisheavy}
	For all $i \in [n]$ we have $Y_i \subset H_i$. Furthermore, for all 
	\begin{itemize}
		\item $g \notin H_i$, and
		\item $g \notin \mathbf{Y}$, 
	\end{itemize}
	we have $v_i(Y_i) \geq v_i(g)$.
\end{lemma}

\begin{proof}
	We first show that $Y_i \subset H_i$. We prove the same by contradiction. Assume otherwise, i.e., $Y_i \not \subset H_i$. In that case note that there is always a good $g \in H_i \setminus \mathbf{Y}$ (as $\lvert H_i \rvert = \lvert \mathbf{Y} \rvert =n$ and there is a good in $\mathbf{Y}$ (namely $Y_i$) which is not in $H_i$). By the definition of $H_i$ it is clear that $\ra_i(g) < \ra_i(g')$ for all $g' \notin H_i$. Thus we have $\ra_i(g) < \ra_i(Y_i)$ when $g \notin \mathbf{Y}$, which is a contradiction. Therefore $Y_i \subset H_i$. This also immediately shows that for all $g \notin H_i$ we have $v_i(g) \leq v_i(Y_i)$ (as $Y_i \subset H_i$ and any good in $H_i$ is at least as valuable as any good outside $H_i$ to agent $i$).

	The proof of the last statement of the lemma is immediate. We have that $\ra_i(Y_i) < \ra_i(g)$ for all $g \notin \mathbf{Y}$, immediately implying that $v_i(Y_i) \geq v_i(g)$.
\end{proof}

\paragraph{Run $(\alpha,c)$-EFX algorithm on the remaining goods:} Once we determined the initial allocation $Y$, we run an $(\alpha,c)$-EFX algorithm on the remaining goods starting with $Y$ as the initial allocation ($Y$ is a feasible initial allocation as it is trivially an $\alpha$-EFX allocation as every agent has exactly a single good). Let $Z$ be the final $(\alpha,c)$-EFX allocation. As mentioned earlier in Section~\ref{overview} the singleton sets allocated to the agents are the barriers to proving our desired approximation for any $(\alpha,c)$-EFX allocation. However since we started with a good initial allocation (namely $Y$), we first show that we have some nice properties about the singleton sets in the final allocation $Z$.

\begin{observation}
	\label{finalsingletons}
	If $ \lvert Z_{\ell} \rvert =1$ for any $\ell \in [n]$, then we have $Z_{\ell} \subset \mathbf{Y}$.
\end{observation}

\begin{proof}
	Since $Z$ is obtained by running an $(\alpha,c)$-EFX algorithm starting with $Y$ as the initial allocation, we have for every agent $i$ that $v_i(Z_i) \geq v_i(Y_i)$ (from the definition of $(\alpha,c)$-EFX algorithm). Note that if for any agent $i$ we have $Z_i \neq Y_i$, and $\lvert Z_i \rvert = 1$, then $v_i(Z_i) > v_i(Y_i)$ (from the definition of $(\alpha,c)$-EFX algorithm). Now consider the agent $\ell$ such that $\lvert Z_{\ell} \rvert =1$. If $Z_{\ell} = Y_{\ell}$, then we immediately have $Z_{\ell} \subset \mathbf{Y}$. So now consider the case when $Z_{\ell} \neq Y_{\ell}$. Then we have $v_{\ell}(Z_{\ell}) > v_{\ell}(Y_{\ell})$. By Lemma~\ref{Yisheavy} we know that no good outside $\mathbf{Y}$ can be more valuable to agent $\ell$ than $Y_{\ell}$. Therefore $Z_{\ell} \subset \mathbf{Y}$. 
\end{proof}

Now we show a lower bound on the final valuation of an agent in terms of the low valued goods.

\begin{observation}
	\label{smallgoodsbound1}
	We have $v_i(Z_i) \geq \frac{\alpha v_i( M \setminus \Y)}{2(n+1)}$ for all $i \in [n]$.
\end{observation}

\begin{proof}
	Fix an agent $i$. Now consider any agent $j$ such that $Z_j$ is not a singleton. Since $Z$ is an $\alpha$-EFX allocation, we have that $v_i(Z_i) \geq \alpha \cdot v_i(Z_j \setminus g)$ for all $g \in Z_j$. Since $\lvert Z_j \rvert  \geq 2$ we can say that $v_i(Z_i) \geq \alpha \cdot \mathit{max}(v_i(Z_j \setminus g),v_i(g))$. Therefore we have,
	
	\begin{align*}
	v_i(Z_i) &\geq \frac{\alpha \cdot (v_i(Z_j \setminus g) + v_i(g))}{2}\\
	&\geq \frac{\alpha \cdot v_i(Z_j)}{2} &(\text{ by subadditivity})
	\end{align*}
	
	Let $S = \cup_{ \lvert Z_{\ell} \rvert =1}Z_{\ell}$. By Observation~\ref{finalsingletons} we know that $S \subseteq \mathbf{Y}$.  Note that even if $Z$ is a partial allocation (if $c=0$ in the $(\alpha,c)$-EFX allocation $Z$) and there exists a set of goods $P$ unallocated, we have $v_i(Z_i) \geq v_i(P)$ (since $Z$ is an $(\alpha,c)$-EFX allocation) . Therefore we have, 
	
	\begin{align*}
	(n+1-\lvert S \rvert)v_i(Z_i) &\geq \tfrac{\alpha}{2} \sum_{ \lvert Z_j \rvert \geq 2} v_i(Z_j) + v_i(P)\\
	&\geq \tfrac{\alpha}{2} v_i \Big(\bigcup_{ \lvert Z_j \rvert \geq 2}  Z_j \cup P \Big) &(\text{ by subadditivity})\\
	&= \tfrac{\alpha}{2}v_i(M \setminus S)\\
	&\geq \tfrac{\alpha}{2}v_i(M \setminus Y) &(\text{ since } S \subseteq \Y) 
	\end{align*}
	
	Therefore we have $v_i(Z_i) \geq \tfrac{\alpha}{2(n+1-\lvert S \rvert)}v_i(M \setminus \Y) \geq \tfrac{\alpha}{2(n+1)}v_i(M \setminus \Y)$. %\JG{We are unncessarily losing a factor of 2 here. I think we should get two bounds with(out) charity. Without charity, we get $\alpha/2n$ and with charity we get $\alpha/2(n+1)$, right? This together with other lemmas, we will manage roughly $4(n+1)$-approximation with charity using $\alpha = (1-\epsilon)$. Let's add these factors exactly for with/without charity in the last theorem. This will also help to show that our bounds are better than the other submission where they have $8n$-approximation!} 
\end{proof}

Now we prove a lower bound on $v_i(Z_i)$ in terms of the initial allocation $Y$ and the set of low valuable goods for agent $i$, i.e., $M \setminus H_i$.

\begin{lemma}
	\label{lowerboundfinal}
	For all $i \in [n]$ we have  $v_i(Z_i) \geq \frac{\alpha}{4(n+1)} \cdot \bigg( n \cdot v_i(Y_i) +   v_i \Big( M \setminus H_i \Big) \bigg)$.
\end{lemma}

\begin{proof}
	We have $v_i(Z_i) \geq v_i(Y_i)$ (since $Z$ is an allocation determined by an $(\alpha,c)$-EFX algorithm with $Y$ as the initial allocation) and from Observation~\ref{smallgoodsbound1} we have $v_i(Z_i) \geq \frac{\alpha v_i( M \setminus \Y)}{2(n+1)}$. Therefore for all $i \in [n]$ we have 
	
	\begin{align*}
	v_i(Z_i) &\geq \frac{1}{2} \cdot \bigg(v_i(Y_i) + \frac{\alpha}{2(n+1)} \cdot v_i( M \setminus \Y) \bigg)\\
	&= \frac{1}{2} \cdot\bigg(v_i(Y_i) + \frac{\alpha}{2(n+1)} \cdot v_i \Big( \big(M \setminus (\Y \cap H_i) \big) \setminus (\Y \setminus H_i) \Big) \bigg)\\
	&\geq\frac{1}{2} \cdot \bigg(v_i(Y_i) + \frac{\alpha}{2(n+1)} \cdot v_i \Big( M \setminus (\Y \cap H_i) \Big) - \frac{\alpha}{2(n+1)} \cdot v_i \Big( \Y \setminus H_i \Big) \bigg) &\text{(by subadditivity)}\\
	&\geq \frac{1}{2} \cdot \bigg(v_i(Y_i) + \frac{\alpha}{2(n+1)} \cdot v_i \Big( M \setminus H_i \Big) - \frac{\alpha}{2(n+1)} \cdot v_i \Big( \Y \setminus H_i \Big) \bigg) &(\text{as } \Y \cap H_i \subseteq H_i) \numberthis \label{boundonZ}
	\end{align*}
	
	By Lemma~\ref{Yisheavy} we know that $v_i(Y_i) \geq v_i(g)$ for all $g \notin H_i$. In particular $v_i(Y_i) \geq v_i(g)$ for all $g \in \Y \setminus H_i$. Thus,
	\begin{align*}
	v_i( \Y \setminus H_i) & \leq  \sum_{i \in \Y \setminus H_i} v_i(g) &(\text{by subadditivity})\\
	& \leq \sum_{i \in \Y \setminus H_i} v_i(Y_i)\\
	& =\lvert \Y \setminus H_i \rvert \cdot v_i(Y_i)\\
	& \leq n \cdot v_i(Y_i) &(\text{as } \lvert \Y \rvert = n)\\ 
	&\leq (n+1) \cdot v_i(Y_i)     
	\end{align*}  
	
	Substituting the upper bound for $v_i(\Y \setminus H_i)$ in (\ref{boundonZ}) we have 
	
	\begin{align*}
	v_i(Z_i) &\geq \frac{1}{2} \cdot \bigg((1 - \frac{\alpha}{2}) \cdot v_i(Y_i) + \frac{\alpha}{2(n+1)} \cdot v_i \Big( M \setminus H_i \Big) \bigg)\\
	&\geq \frac{1}{2} \cdot \bigg( \frac{1}{2} \cdot v_i(Y_i) + \frac{\alpha}{2(n+1)} \cdot v_i \Big( M \setminus H_i \Big) \bigg)  &(\text{as } \alpha \leq 1)\\  
	&= \frac{\alpha}{4(n+1)} \cdot \bigg( (n+1) \cdot \frac{1}{\alpha} \cdot v_i(Y_i) +   v_i \Big( M \setminus H_i \Big) \bigg)\\
	&\geq \frac{\alpha}{4(n+1)} \cdot \bigg( n \cdot v_i(Y_i) +   v_i \Big( M \setminus H_i \Big) \bigg)  \qedhere
	\end{align*}
\end{proof}

\begin{algorithm}[!t]
	\caption {Determining an $(\alpha,c)$-EFX allocation with ${\mathcal{O}}(n)$ approximation on optimum $p$-mean.}
	\label{algorithmpmeans}
	\begin{algorithmic}[1]
		%\State \textbf{Determine} $B$.   
		\State \textbf{Construct} $G = \langle[n] \cup M, [n] \times M \rangle$ with \[
		w_{ig}=
		\begin{cases}
		n \cdot v_i(g) + v_i(M \setminus H_i) & \text{if $p = - \infty$} \\
		\log \Big(n \cdot v_i(g) + v_i(M \setminus H_i) \Big) &\text{if $p = 0$}\\
		\Big(n \cdot v_i(g) + v_i(M \setminus H_i)\Big)^p & \text{otherwise} 
		\end{cases}
		\] 
		\State \textbf{Set} $Y$ such that \[
		\cup_{i \in [n]} (i,Y_i)=
		\begin{cases}
		\text{\textbf{Max-Min-Matching}$(G)$} & \text{if $p = - \infty$} \\
		\text{\textbf{Min-Weight-Perfect-Matching}$(G)$ } & \text{if $p < 0$ and $p \neq -\infty$} \\
		\text{\textbf{Max-Weight-Matching}$(G)$} & \text{otherwise} \\ 
		\end{cases}
		\] 
		\While{$\exists i \in [n]$ and $\exists g \notin \mathbf{Y}$ such that $\ra_i(g) < \ra_i(Y_i)$}
		\State $Y_i \gets \left\{g\right\}$.
		\EndWhile   
		
		\State \textbf{Set} $Z \gets \mathbf{(\alpha,c)}$-\textbf{EFX}$\Big(\langle Y_1, \dots ,Y_n \rangle, \big( M \setminus \cup_{i \in [n]} Y_i \big) \Big)$
	\end{algorithmic}
\end{algorithm}

The final allocation is the set $Z$ which is obtained by running an $(\alpha,c)$-EFX allocation starting with $Y$ as the initial allocation. Therefore our final allocation is an $(\alpha,c)$-EFX allocation. We would now show the approximation guarantees that the algorithm achieves. The sections that follow prove that the allocation $Z$ has a $p$-welfare that is $\tfrac{\alpha}{4(n+1)}$ times  $p$-mean welfare achieved by the optimal allocation. Each section from here on presents the proof for particular value or a range of values of  $p$.

\subsection{Case $p = -\infty$}
This is the case where $M_p(X) = \mathit{min}_{i \in [n]} v_i(X_i)$. Let $X^*$ be the allocation with the highest $p$-mean value and let $g^*_i$ be agent $i$'s most valuable good in $X^*_i$. We will show in this section that $M_p(Z) \geq \frac{\alpha}{4(n+1)} \cdot M_p(X^*)$. First observe that by Lemma~\ref{lowerboundfinal}, we have that for all $i \in [n]$, $v_i(Z_i) \geq \frac{\alpha}{4(n+1)} \cdot \bigg( n \cdot v_i(Y_i) +   v_i \Big( M \setminus H_i \Big) \bigg)$. Therefore,

\begin{align*}
\mathit{min}_{i \in [n]}v_i(Z_i) &\geq  \mathit{min}_{i \in [n]} \frac{\alpha}{4(n+1)} \cdot \bigg( n \cdot v_i(Y_i) +   v_i \Big( M \setminus H_i \Big) \bigg)
\end{align*}

Recall that  $Y$ was chosen such that $(i,Y_i)$ is a maximum weight matching in the bipartite graph $G = ([n] \cup M, [n] \times M)$ where the weight of an edge from agent $i$ to good $g$, $w_{ig} = n \cdot v_i(g) + v_i(M \setminus H_i)$. Also note that $\cup_{i \in [n]} (i,g^*_i)$ is a feasible matching in $G$. Thus we have 

\begin{align*}
\mathit{min}_{i \in [n]} \Big( n \cdot v_i(Y_i) +   v_i (M \setminus H_i) \Big) &\geq \mathit{min}_{i \in [n]} \Big( n \cdot v_i(g^*_i) +   v_i (M \setminus H_i) \Big)
\end{align*}

Therefore we have,

\begin{align*}
\mathit{min}_{i \in [n]}v_i(Z_i) &\geq  \mathit{min}_{i \in [n]} \frac{\alpha}{4(n+1)} \cdot \Big( n \cdot v_i(Y_i) +   v_i ( M \setminus H_i ) \Big)\\
&\geq  \frac{\alpha}{4(n+1)} \cdot \mathit{min}_{i \in [n]} \Big( n \cdot v_i(g^*_i) +   v_i (M \setminus H_i) \Big)\\
&\geq  \frac{\alpha}{4(n+1)} \cdot \mathit{min}_{i \in [n]} \bigg( n \cdot v_i(g^*_i) +   v_i \Big(X^*_i \cap (M \setminus H_i) \Big) \bigg)\\
&\geq  \frac{\alpha}{4(n+1)} \cdot \mathit{min}_{i \in [n]} \bigg( v_i(X^*_i \cap H_i) +   v_i \Big(X^*_i \cap (M \setminus H_i) \Big) \bigg) &(\text{as } \lvert H_i \rvert =n)\\
&\geq  \frac{\alpha}{4(n+1)} \cdot \mathit{min}_{i \in [n]} v_i(X^*_i) &\text{(by subadditivity)}
\end{align*}

This shows that $M_p(Z) \geq \frac{\alpha}{4(n+1)} \cdot M_p(X^*)$ when $p = -\infty$.

\subsection{Case $p < 0$ and $p \neq -\infty$}
\label{negativep}
The proof in this section is very similar to the proof when $p=-\infty$. Still for completeness we sketch the whole proof. Let $X^*$ be the allocation with the highest $p$-mean value and let $g^*_i$ be agent $i$'s most valuable good in $X^*_i$. Similar to the case $p = - \infty$, will show in this section that $M_p(Z) \geq \frac{\alpha}{4(n+1)} \cdot M_p(X^*)$. We now define 

\begin{align*}
R(Z)&= \sum_{i \in [n]} v_i(Z_i)^p 
\end{align*}

Note that $M_p(Z) = \Big( \tfrac{1}{n} \cdot  R(Z) \Big)^{\tfrac{1}{p}}$. We now prove an upper bound on $R(Z)$.

\begin{lemma}
	\label{boundonR}
	We have $R(Z) \leq \frac{\alpha^p}{\big( 4(n+1) \big)^p} \cdot \Big( \sum_{i \in [n]} v_i(X^*_i)^p \Big)$.
\end{lemma}

By Lemma~\ref{lowerboundfinal}, we have that for all $i \in [n]$, $v_i(Z_i) \geq \frac{\alpha}{4(n+1)} \cdot \bigg( n \cdot v_i(Y_i) +   v_i \Big( M \setminus H_i \Big) \bigg)$. Therefore,

\begin{align*}
R(Z) &\leq \sum_{i \in [n]} \bigg(\frac{\alpha}{4(n+1)} \cdot \Big( n \cdot v_i(Y_i) +  v_i(M \setminus H_i) \Big) \bigg)^p &(\text{as $p$ is negative})
\end{align*}

Recall that  $Y$ was chosen such that $(i,Y_i)$ is a minimum  weight perfect matching in the bipartite graph $G = ([n] \cup M, [n] \times M)$ where the weight of an edge from agent $i$ to good $g$, $w_{ig} =  \Big(n \cdot v_i(g) + v_i(M \setminus H_i) \Big)^p$. Note that $\cup_{i \in [n]} (i,g^*_i)$ is a feasible matching in $G$. Thus we have, 

\begin{align*}
\sum_{i \in [n]} \Big( n \cdot v_i(Y_i) +  v_i(M \setminus H_i) \Big)^p &\leq \sum_{i \in [n]} \Big( n \cdot v_i(g^*_i) +  v_i(M \setminus H_i) \Big)^p 
\end{align*}

Therefore we have\footnote{For the set of inequalities that follow the reader is reminded that we are in the case where $p<0$.}

\begin{align*}
R(Z) &\leq \sum_{i \in [n]} \bigg(\frac{\alpha}{4(n+1)} \cdot \Big(n \cdot  v_i(Y_i) +  v_i(M \setminus H_i) \Big) \bigg)^p\\
&\leq \sum_{i \in [n]} \bigg(\frac{\alpha}{4(n+1)} \cdot \Big(n \cdot  v_i(g^*_i) +  v_i(M \setminus H_i) \Big) \bigg)^p\\
&= \frac{\alpha^p}{\big( 4(n+1) \big)^p} \cdot \sum_{i \in [n]} \Big(n \cdot  v_i(g^*_i) +  v_i(M \setminus H_i) \Big)^p\\
&\leq \frac{\alpha^p}{\big( 4(n+1) \big)^p} \cdot \sum_{i \in [n]} \Big(n \cdot  v_i(g^*_i) +  v_i \big(X^*_i \cap (M \setminus H_i) \big) \Big)^p\\     
&\leq \frac{\alpha^p}{\big( 4(n+1) \big)^p} \cdot \sum_{i \in [n]} \Big( v_i(X^*_i \cap H_i) +  v_i\big(X^*_i \cap (M \setminus H_i) \big) \Big)^p  &(\text{as $\lvert H_i \rvert =n$})\\
&\leq \frac{\alpha^p}{\big( 4(n+1) \big)^p} \cdot \sum_{i \in [n]} v_i(X^*_i)^p  &\text{(by subadditivity)}
\end{align*}

Now we are ready to prove the guarantee on the $p$-mean welfare. We have, 
\begin{align*}
M_p(Z) &= \Big(\frac{1}{n} \cdot R(Z) \Big)^{\tfrac{1}{p}}\\
&\geq \bigg ( \frac{1}{n} \cdot \frac{\alpha^p}{\big( 4(n+1) \big)^p} \cdot \Big( \sum_{i \in [n]} v_i(X^*_i)^p \Big) \bigg)^{\tfrac{1}{p}} \text{ (by Lemma~\ref{boundonR} and also p is negative)}\\
&\geq \frac{\alpha}{4(n+1)} \cdot M_p(X^*) \qedhere
\end{align*}

\subsection{Case $p=0$: Nash Welfare}
\label{Nash welfare}
This is the case where $M_p(X) = \Big(\prod_{i \in [n]} v_i(X_i) \Big)^{\tfrac{1}{n}}$. Let $X^*$ be the allocation with the highest $p$-mean value and let $g^*_i$ be agent $i$'s most valuable good in $X^*_i$. Like in the earlier sections we will show in this section that $M_p(Z) \geq \frac{\alpha}{4(n+1)} \cdot M_p(X^*)$. First observe that by Lemma~\ref{lowerboundfinal}, we have that for all $i \in [n]$, $v_i(Z_i) \geq \frac{\alpha}{4(n+1)} \cdot \bigg( n \cdot v_i(Y_i) +   v_i \Big( M \setminus H_i \Big) \bigg)$. Therefore,

\begin{align*}
\bigg(\prod_{i \in [n]} v_i(Z_i) \bigg)^{\tfrac{1}{n}} &\geq \bigg(\prod_{i \in [n]} \frac{\alpha}{4(n+1)} \cdot \bigg( n \cdot v_i(Y_i) +   v_i \Big( M \setminus H_i \Big) \bigg)^{\tfrac{1}{n}}\\
&=  \frac{\alpha}{4(n+1)} \cdot \bigg(\prod_{i \in [n]} \bigg( n \cdot v_i(Y_i) +   v_i \Big( M \setminus H_i \Big) \bigg)^{\tfrac{1}{n}}                                                       
\end{align*}

Recall that  $Y$ was chosen such that $(i,Y_i)$ is a maximum weight matching in the bipartite graph $G = ([n] \cup M, [n] \times M)$ where the weight of an edge from agent $i$ to good $g$, $w_{ig} = \log \Big(n \cdot v_i(g) + v_i(M \setminus H_i) \Big)$. Note that $\cup_{i \in [n]} (i,g^*_i)$ is a feasible matching in $G$. Thus we have 

\begin{align*}
\sum_{i \in [n]} \log \Big(n \cdot v_i(Y_i) + v_i(M \setminus H_i) \Big) &\geq \sum_{i \in [n]} \log \Big(n \cdot v_i(g^*_i) + v_i(M \setminus H_i) \Big)\\
\implies \prod_{i \in [n]} \Big(n \cdot v_i(Y_i) + v_i(M \setminus H_i) \Big) &\geq  \prod_{i \in [n]} \Big(n \cdot v_i(g^*_i) + v_i(M \setminus H_i) \Big)
\end{align*}
Therefore we have,

\begin{align*}
\bigg( \prod_{i \in [n]} v_i(Z_i) \bigg)^{\tfrac{1}{n}} &\geq \frac{\alpha}{4(n+1)} \cdot \bigg( \prod_{i \in [n]} \Big( n \cdot v_i(Y_i) + v_i(M \setminus H_i) \Big)  \bigg)^{\tfrac{1}{n}}\\
&\geq \frac{\alpha}{4(n+1)} \cdot \bigg( \prod_{i \in [n]} \Big(n \cdot v_i(g^*_i) + v_i(M \setminus H_i) \Big) \bigg)^{\tfrac{1}{n}}\\
&\geq \frac{\alpha}{4(n+1)} \cdot \bigg( \prod_{i \in [n]} \Big(n \cdot v_i(g^*_i) + v_i\big(X^*_i \cap (M \setminus H_i) \big) \Big) \bigg)^{\tfrac{1}{n}}\\
&\geq \frac{\alpha}{4(n+1)} \cdot \bigg( \prod_{i \in [n]} \Big(v_i(X^*_i \cap H_i) + v_i\big(X^*_i \cap (M \setminus H_i) \big) \Big) \bigg)^{\tfrac{1}{n}} &(\text{as } \lvert H_i \rvert =n)\\
&\geq \frac{\alpha}{4(n+1)} \cdot \Big( \prod_{i \in [n]} v_i(X^*_i) \Big)^{\tfrac{1}{n}} &(\text{ by subadditivity})
\end{align*}
This shows that $M_p(Z) \geq \frac{\alpha}{4(n+1)} \cdot M_p(X^*)$ when $p = 0$.

\subsection{Case $p \in (0,1]$}
The proof of the approximation guarantee in this case follows almost the same proof in the Section~\ref{negativep}, with the only difference that since $p$ is positive and we compute a Maximum weight matching in the bipartite graph $G = ([n] \cup M, [n] \times M)$ where the weight of an edge from agent $i$ to good $g$, $w_{ig} =  \Big(n \cdot v_i(g) + v_i(M \setminus H_i) \Big)^p$ and  we will have lower bounds on $R(Z)$ and consequently also lower bounds on $M_p(Z)$.

Therefore our algorithm computes an $(\alpha,c)$-EFX allocation which is also an $\tfrac{4(n+1)}{\alpha}$ approximation of the optimum $p$-mean welfare.

\begin{theorem}
	\label{mainthm}
	Given any instance $\langle [n],M, \mathcal{V} \rangle$, in polynomial time we can determine an allocation $Z$ such that 
	\begin{itemize}
		\item $Z$ is either $(1-\varepsilon,0)$-EFX allocation or $(\tfrac{1}{2}- \varepsilon,1)$-EFX allocation for any positive $\varepsilon$ and 
		\item $M_p(Z) \geq \tfrac{\alpha}{4(n+1)}M_p(X^*)$. %\JG{The factor $\tfrac{\alpha}{4(n+1)}$ is missing.}
	\end{itemize}
	where $X^*$ is the allocation with maximum $p$-mean welfare.
\end{theorem}

\begin{proof}
	We showed that the allocation $Z$ computed by Algorithm~\ref{algorithmpmeans} is an $(\alpha,c)$-EFX allocation and $M_p(Z) \geq \tfrac{\alpha}{4(n+1)} \cdot M_p (X^*)$. It suffices to show that Algorithm~\ref{algorithmpmeans} runs in polynomial time. Note that steps 1 of the algorithm can be implemented in $\mathit{poly(n,m)}$ time. Step 2 can also be realized in polynomial time as all the matching subroutines run in $\mathit{poly}(n,m)$. The while loop in step 3 runs for $\mathit{poly}(n,m)$ iterations as with each iterations $\sum_{i \in [n]} \ra_i(Y_i)$ decreases by 1 and $n < \sum_{i \in [n]} \ra_i(Y_i) \leq nm$. In step 4, we run the $(\alpha,c)$-EFX algorithm with $Y$ as the initial allocation. Plaut and Roughgarden [PR18] and Chaudhury et al. [CKMS20] show an $(\tfrac{1}{2} -\varepsilon,1)$-EFX algorithm and $(1-\varepsilon,0)$-EFX algorithm respectively that runs in $\mathit{poly}(n,m,\tfrac{1}{\varepsilon})$ time. Therefore we can obtain an allocation $Z$ with the properties mentioned in theorem in $\mathit{poly}(n,m,\tfrac{1}{\varepsilon})$ time.  
\end{proof}

\textbf{Remark:} Theorem~\ref{mainthm} also suggest that we can find a $\frac{4(n+1)}{1 -\varepsilon}$ approximation to the $p$-mean welfare in polynomial time. Also it can be verified that a minor variant of our approach (changing the weights of the edges of the complete bipartite graph $G([n] \cup B, [n] \times B)$ appropriately - step 1 of Algorithm~\ref{algorithmpmeans}) gives a $\mathcal{O}(n)$ approximation on weighted generalized $p$-mean, defined as $WM_p(X) =   \big( \sum_{i \in [n]} \eta_i \cdot v_i(X_i)^p \big)^{\tfrac{1}{p}}$. In particular, we also get an $\mathcal{O}(n)$ approximation algorithm for asymmetric Nash welfare when agents have submodular valuations (improving the current best bound of $\mathcal{O}(n \cdot \log n)$ by Garg et al.~\cite{GargKK20}).

%\bibliographystyle{alpha}
%\bibliography{EFX}

\begin{thebibliography}{AMGV18}
	
	\bibitem[ABF{\etalchar{+}}20]{ABFHV20}
	Georgios Amanatidis, Georgios Birmpas, Aris Filos{-}Ratsikas, Alexandros
	Hollender, and Alexandros~A. Voudouris.
	\newblock Maximum nash welfare and other stories about {EFX}.
	\newblock {\em CoRR}, abs/2001.09838, 2020.
	
	\bibitem[AGSS17]{AnariGSS17}
	Nima Anari, Shayan~Oveis Gharan, Amin Saberi, and Mohit Singh.
	\newblock {Nash Social Welfare, Matrix Permanent, and Stable Polynomials}.
	\newblock In {\em 8th Innovations in Theoretical Computer Science Conference
		(ITCS)}, pages 1--12, 2017.
	
	\bibitem[AMGV18]{AnariMGV18}
	Nima Anari, Tung Mai, Shayan~Oveis Gharan, and Vijay~V. Vazirani.
	\newblock Nash social welfare for indivisible items under separable,
	piecewise-linear concave utilities.
	\newblock In {\em Proc.\ 29th Symp.\ Discrete Algorithms (SODA)}, pages
	2274--2290, 2018.
	
	\bibitem[AMNS17]{AMNS17}
	Georgios Amanatidis, Evangelos Markakis, Afshin Nikzad, and Amin Saberi.
	\newblock Approximation algorithms for computing maximim share allocations.
	\newblock {\em ACM Transactions on Algorithms}, 13(4):52:1--52:28, 2017.
	
	\bibitem[BBKS20]{BarmanBKS'20}
	Siddharth Barman, Umang Bhaskar, Anand Krishna, and Ranjani~G. Sundaram.
	\newblock Tight approximation algorithms for p-mean welfare under subadditive
	valuations.
	\newblock {\em CoRR}, abs/2005.07370, 2020.
	
	\bibitem[BBMN18]{BarmanBMN18}
	Siddharth Barman, Arpita Biswas, Sanath Kumar~Krishna Murthy, and Yadati
	Narahari.
	\newblock Groupwise maximin fair allocation of indivisible goods.
	\newblock In {\em {AAAI}}, pages 917--924. {AAAI} Press, 2018.
	
	\bibitem[BCKO17]{BudishCKO17}
	Eric Budish, G{\'{e}}rard~P. Cachon, Judd~B. Kessler, and Abraham Othman.
	\newblock Course match: {A} large-scale implementation of approximate
	competitive equilibrium from equal incomes for combinatorial allocation.
	\newblock {\em Operations Research}, 65(2):314--336, 2017.
	
	\bibitem[BK17]{BK17}
	Siddharth Barman and Sanath~Kumar Krishnamurthy.
	\newblock Approximation algorithms for maximin fair division.
	\newblock In {\em Proceedings of the 18th ACM Conference on Economics and
		Computation (EC)}, pages 647--664, 2017.
	
	\bibitem[BK19]{BarmanK19}
	Siddharth Barman and Sanath~Kumar Krishnamurthy.
	\newblock On the proximity of markets with integral equilibria.
	\newblock In {\em Proc.\ 33rd Conf.\ Artif.\ Intell.\ (AAAI)}, 2019.
	
	\bibitem[BKV18]{BKV18}
	Siddharth Barman, Sanath~Kumar Krishnamurthy, and Rohit Vaish.
	\newblock Finding fair and efficient allocations.
	\newblock In {\em Proceedings of the 19th ACM Conference on Economics and
		Computation (EC)}, pages 557--574, 2018.
	
	\bibitem[BL16]{BL16}
	Sylvain Bouveret and Michel Lema\^itre.
	\newblock Characterizing conflicts in fair division of indivisible goods using
	a scale of criteria.
	\newblock In {\em Autonomous Agents and Multi-Agent Systems (AAMAS) 30, 2},
	pages 259--290, 2016.
	
	\bibitem[BS20]{BarmanRanjani_pmeans}
	Siddharth Barman and Ranjani~G. Sundaram.
	\newblock Uniform welfare guarantees under identical subadditive valuations.
	\newblock {\em CoRR}, abs/2005.00504, 2020.
	
	\bibitem[Bud11]{budish2011combinatorial}
	Eric Budish.
	\newblock The combinatorial assignment problem: Approximate competitive
	equilibrium from equal incomes.
	\newblock {\em Journal of Political Economy}, 119(6):1061--1103, 2011.
	
	\bibitem[CCG{\etalchar{+}}18]{ChaudhuryCGGHM18}
	Bhaskar~Ray Chaudhury, Yun~Kuen Cheung, Jugal Garg, Naveen Garg, Martin Hoefer,
	and Kurt Mehlhorn.
	\newblock On fair division for indivisible items.
	\newblock In {\em 38th {IARCS} Annual Conference on Foundations of Software
		Technology and Theoretical Computer Science, {FSTTCS}}, pages 25:1--25:17,
	2018.
	
	\bibitem[CDG{\etalchar{+}}17]{ColeDGJMVY17}
	Richard Cole, Nikhil Devanur, Vasilis Gkatzelis, Kamal Jain, Tung Mai, Vijay
	Vazirani, and Sadra Yazdanbod.
	\newblock Convex program duality, {F}isher markets, and {N}ash social welfare.
	\newblock In {\em Proc.\ 18th Conf.\ Economics and Computation (EC)}, 2017.
	
	\bibitem[CFS17]{ConitzerFS17}
	Vincent Conitzer, Rupert Freeman, and Nisarg Shah.
	\newblock Fair public decision making.
	\newblock In {\em Proc.\ 18th Conf.\ Economics and Computation (EC)}, pages
	629--646, 2017.
	
	\bibitem[CG18]{ColeG18}
	Richard Cole and Vasilis Gkatzelis.
	\newblock Approximating the nash social welfare with indivisible items.
	\newblock {\em {SIAM} J. Comput.}, 47(3):1211--1236, 2018.
	
	\bibitem[CGH19]{CaragiannisGravin19}
	Ioannis Caragiannis, Nick Gravin, and Xin Huang.
	\newblock Envy-freeness up to any item with high {N}ash welfare: The virtue of
	donating items.
	\newblock In {\em {EC}}, pages 527--545. {ACM}, 2019.
	
	\bibitem[CKM{\etalchar{+}}16]{CaragiannisKMP016}
	Ioannis Caragiannis, David Kurokawa, Herv{\'{e}} Moulin, Ariel~D. Procaccia,
	Nisarg Shah, and Junxing Wang.
	\newblock The unreasonable fairness of maximum {Nash} welfare.
	\newblock In {\em Proceedings of the 17th ACM Conference on Economics and
		Computation (EC)}, pages 305--322, 2016.
	
	\bibitem[CKMS19]{CKMSarxiv}
	Bhaskar~Ray Chaudhury, Telikepalli Kavitha, Kurt Mehlhorn, and Alkmini
	Sgouritsa.
	\newblock A little charity guarantees almost envy-freeness.
	\newblock {\em CoRR}, abs/1907.04596, 2019.
	
	\bibitem[CKMS20]{CKMS20}
	Bhaskar~Ray Chaudhury, Telikepalli Kavitha, Kurt Mehlhorn, and Alkmini
	Sgouritsa.
	\newblock A little charity guarantees almost envy-freeness.
	\newblock In {\em Proceedings of the 31st Symposium on Discrete Algorithms
		(SODA)}, pages 2658--2672, 2020.
	
	\bibitem[GHM18]{GargHM18}
	Jugal Garg, Martin Hoefer, and Kurt Mehlhorn.
	\newblock Approximating the {N}ash social welfare with budget-additive
	valuations.
	\newblock In {\em Proc.\ 29th Symp.\ Discrete Algorithms (SODA)}, 2018.
	
	\bibitem[GHS{\etalchar{+}}18]{GhodsiHSSY18}
	Mohammad Ghodsi, Mohammad~Taghi Hajiaghayi, Masoud Seddighin, Saeed Seddighin,
	and Hadi Yami.
	\newblock Fair allocation of indivisible goods: Improvements and
	generalizations.
	\newblock In {\em Proceedings of the 2018 {ACM} Conference on Economics and
		Computation (EC)}, pages 539--556, 2018.
	
	\bibitem[GKK20]{GargKK20}
	Jugal Garg, Pooja Kulkarni, and Rucha Kulkarni.
	\newblock Approximating {N}ash social welfare under submodular valuations
	through (un)matchings.
	\newblock In {\em SODA}, 2020.
	
	\bibitem[GM19]{GargM19}
	Jugal Garg and Peter McGlaughlin.
	\newblock Improving {N}ash social welfare approximations.
	\newblock In {\em {IJCAI}}, pages 294--300. ijcai.org, 2019.
	
	\bibitem[GMT19]{JGargMT19}
	Jugal Garg, Peter McGlaughlin, and Setareh Taki.
	\newblock Approximating maximin share allocations.
	\newblock In {\em Proceedings of the 2nd Symposium on Simplicity in Algorithms
		(SOSA)}, volume~69, pages 20:1--20:11. Schloss Dagstuhl - Leibniz-Zentrum
	fuer Informatik, 2019.
	
	\bibitem[GT19]{GargT19}
	Jugal Garg and Setareh Taki.
	\newblock An improved approximation algorithm for maximin shares.
	\newblock {\em CoRR}, abs/1903.00029, 2019.
	
	\bibitem[KP07]{KhotPonuswami07}
	Subhash Khot and Ashok~Kumar Ponnuswami.
	\newblock Approximation algorithms for the max-min allocation problem.
	\newblock In {\em {APPROX-RANDOM}}, volume 4627 of {\em Lecture Notes in
		Computer Science}, pages 204--217. Springer, 2007.
	
	\bibitem[KPW18]{KPW18}
	David Kurokawa, Ariel~D. Procaccia, and Junxing Wang.
	\newblock Fair enough: Guaranteeing approximate maximin shares.
	\newblock {\em Journal of ACM}, 65(2):8:1--27, 2018.
	
	\bibitem[Lee17]{Lee17}
	Euiwoong Lee.
	\newblock {APX}-hardness of maximizing {N}ash social welfare with indivisible
	items.
	\newblock {\em Inf. Process. Lett.}, 122:17--20, 2017.
	
	\bibitem[LLN06]{LehmannLN06}
	Benny Lehmann, Daniel Lehmann, and Noam Nisan.
	\newblock Combinatorial auctions with decreasing marginal utilities.
	\newblock {\em Games Econ. Behav.}, 55(2):270--296, 2006.
	
	\bibitem[LMMS04]{LiptonMMS04}
	Richard~J. Lipton, Evangelos Markakis, Elchanan Mossel, and Amin Saberi.
	\newblock On approximately fair allocations of indivisible goods.
	\newblock In {\em Proceedings of the 5th ACM Conference on Electronic Commerce
		(EC)}, pages 125--131, 2004.
	
	\bibitem[NR14]{NguyenR14}
	Trung~Thanh Nguyen and J{\"o}rg Rothe.
	\newblock Minimizing envy and maximizing average {N}ash social welfare in the
	allocation of indivisible goods.
	\newblock {\em Discrete Applied Mathematics}, 179:54--68, 2014.
	
	\bibitem[PR18]{TimPlaut18}
	Benjamin Plaut and Tim Roughgarden.
	\newblock Almost envy-freeness with general valuations.
	\newblock In {\em Proceedings of the 29th Annual ACM-SIAM Symposium on Discrete
		Algorithms (SODA)}, pages 2584--2603, 2018.
	
	\bibitem[Ste48]{Steinhaus48}
	Hugo Steinhaus.
	\newblock The problem of fair division.
	\newblock {\em Econometrica}, 16(1):101--104, 1948.
	
\end{thebibliography}

\newcommand{\etalchar}[1]{$^{#1}$}

\end{document}